\documentclass[letterpaper, 10 pt, conference]{ieeeconf}  

\IEEEoverridecommandlockouts                              

\usepackage{graphics} 
\usepackage{graphicx}
\usepackage{color}

\usepackage{mathptmx} 
\usepackage{times} 
\usepackage{amsmath} 
\usepackage{amssymb}  
\usepackage[mathscr]{euscript}
\usepackage{bbm}
\usepackage{subfigure}
\usepackage{multirow}
\usepackage[english]{babel}
\usepackage{booktabs}
\newtheorem{prop}{Proposition}

\usepackage{float}
\usepackage{comment}

\usepackage{bm}

\newcommand{\Et}{E_{\text{time}}}
\newcommand{\Eth}{\hat{E}_{\text{time}}}
\newcommand{\El}{E_{\text{lag.}}}
\newcommand{\Elh}{\hat{E}_{\text{lag.}}}
\newcommand{\Yt}{Y_{\text{time}}}
\newcommand{\Yl}{Y_{\text{lag.}}}
\newcommand{\Ylb}{\bar{Y}_{\text{lag.}}}
\newcommand{\Ylh}{\hat{Y}_{\text{lag.}}}
\newcommand{\Ylhs}{{\hat{Y}}^\star_{\text{lag.}}}

\newcommand{\Ul}{U_{\text{lag.}}}

\newcommand{\St}{\Sigma_{\text{time}}}
\newcommand{\Sl}{\Sigma_{\text{lag.}}}
\renewcommand{\l}{\bm {\ell}}

\usepackage{accents}
\newlength{\dhatheight}
\newcommand{\doublehat}[1]{%
    \settoheight{\dhatheight}{\ensuremath{\hat{#1}}}%
    \addtolength{\dhatheight}{-0.15ex}%
    \hat{\vphantom{\rule{1pt}{\dhatheight}}%
    \smash{\hat{#1}}}}
    
\renewcommand{\doublehat}[1]{\breve{#1}}
    
\title{\LARGE \bf
 Noise reduction in Laguerre-domain discrete delay estimation
}

\author{Mohamed  Abdalmoaty and Alexander Medvedev
\thanks{M. Abdalmoaty ({\small mohamed.abdalmoaty@it.uu.se}) and A. Medvedev ({\small alexander.medvedev@it.uu.se}) are with the Division of Systems and Control, Department of Information Technology, Uppsala University,
				751 05 Uppsala, Sweden}%
}

\begin{document}

\maketitle
\thispagestyle{empty}
\pagestyle{empty}

\begin{abstract}
This paper introduces a stochastic framework for a recently proposed discrete-time delay estimation method in Laguerre-domain, i.e. with the delay block input and output signals being represented by the corresponding Laguerre series. A novel Laguerre domain disturbance model is devised, which allows the involved signals to be square-summable sequences and is suitable in a number of important applications. The relation to two commonly used time-domain disturbance models is clarified. Furthermore, by forming the input signal in a certain way, the signal shape of an additive output disturbance can be estimated and utilized for noise reduction. It is demonstrated that a significant improvement in the delay estimation error is achieved when the noise sequence is correlated. The noise reduction approach is applicable to other Laguerre-domain problems than pure delay estimation.
\end{abstract}

\section{Introduction}

Delays, also termed as dead time, time lag, latency, etc., are ubiquitous in real world and have to be quantified to be properly taken into account in control \cite{Fridman:2004} or estimation applications. Pure delay estimation is instrumental in the remote sensing technology such as radar, sonar, ultrasound and lidar \cite{M92}. All these methods are essentially based on time-of-arrival estimation and make use of emitted pulses with finite support.

Laguerre functions are traditionally used for representing both dynamical systems and signals. In the former case, they are used to approximate the input-output mapping of the system, and, in the latter, to capture the signal forms of the involved inputs and outputs. In time-domain, the  Laguerre functions are essentially exponentials with polynomial coefficients and, therefore, are highly suitable for describing solutions of linear time-invariant systems \cite{W91}.

The idea of using Laguerre functions to estimate delay in continuous- and discrete-time has been investigated before (see, e.g. \cite{IHD01}, \cite{FM99}) but the discrete-time case has drawn less attention until recently. An extensive comparison of delay estimation approaches in simulation experiments was performed in \cite{BL03} and highlighted the robustness of the methods based on the use of Laguerre functions.

The contribution of the present work is threefold: First, a novel disturbance model constituting of a linear combination of a finite number of Laguerre functions with random weights is introduced. Second, an approach to reconstructing the signal shape of additive measurement disturbances through shaping the excitation in Laguerre-domain is proposed. The signal shape reconstruction appears to be most accurate when used on data corrupted by noise generated by the novel disturbance model.
Third, the performance of a time-delay estimation algorithm is improved by applying  Laguerre-domain noise reduction making use of the signal shape reconstruction.

The rest of the paper is organized as follows. After summarizing the necessary background on Laguerre-domain system representation, three stochastic disturbance models are formulated. Further,  making use of the considered models, the impact of disturbance on the estimated Laguerre spectrum of the measured output is analyzed. An approach to reconstructing the signal shape of an additive measurement disturbance realization is described and analyzed for the considered noise models. The disturbance estimate is shown to be instrumental in noise reduction of Laguerre-domain estimation algorithms.
Finally, the efficacy of the proposed noise reduction method is demonstrated on a Laguerre-domain time-delay estimation algorithm via numerical experiments. 

\section{System description}
Consider  the pure discrete-time delay in the face of disturbance\vspace{-.3cm}
\begin{equation}\label{eq:ddelay}
y(t)= u(t-\tau)+e(t), \quad t\in\mathbb{N}_0,
\end{equation}
where $u(t), y(t) \in \mathbb{R}$,  $\tau\in\mathbb{N}_+$ is a constant delay, and $e(t)$ represents unknown noise\footnote{In this paper we use ``noise" and ``disturbance" interchangeably}.

\subsection{Laguerre spectrum}

Let $\mathbb{H}_d^2$ be the Hardy space of analytic functions on the complement of the unit disc that are square-integrable on the unit circle and equipped with  the inner product 
\begin{align}\label{eq:inner_disc}
\langle W, V \rangle = \frac{1}{2\pi i}\oint_{D} \ W(z)V(z^{-1})~\frac{{\mathrm d}z}{z},
\end{align}
where  $D$ is  the unit circle. An orthonormal complete basis in  $\mathbb{H}^2_d$ is given by the  discrete Laguerre functions specified in  $\mathscr Z$-domain  by
\begin{equation}\label{eq:Laguerre}
    L_j(z;p)=\frac{\sqrt{1-p}}{z-\sqrt{p}}T^j(z;p), \quad T(z;p)\triangleq \frac{1-\sqrt{p}z}{z-\sqrt{p}}, 
\end{equation}
for all $j \in \mathbb{N}$, where the constant $0<p<1$ is the discrete Laguerre parameter. Then,  any function $W \in \mathbb{H}^2_d$  can be represented as an infinite series
\begin{equation}\label{eq:series}
W(z)=\sum_{k=0}^\infty w_k L_k(z;p),  \quad  w_j=\langle W, L_j \rangle ,
\end{equation} 
and the set $\{w_j\}_{j\in\mathbb{N}}$ is referred to as the {\em Laguerre spectrum} of $W$.
A system is said to be considered in Laguerre-domain when its inputs and outputs  are given by their Laguerre spectra.

The time-domain representations of the Laguerre functions $\ell_j(t;p)=\mathscr{Z}^{-1} \left\{L_j(z;p)\right\}$ ($j\in\mathbb{N}$) yield an orthonormal basis in $ \l^2\lbrack 0,\infty)$, the space of square-summable sequences defined for non-negative integer arguments, where $\mathscr{Z}^{-1}$ denotes the inverse $\mathscr{Z}$-transform.

\subsection{Linear time-invariant system in Laguerre-domain}
Consider the linear time-invariant (LTI) system with 
\begin{align}
    x(t+1)&=Ax(t)+Bu(t), \ \forall t \in \mathbb{N}_0, \label{eq:LTI} \\
    y(t)&=Cx(t), \nonumber
\end{align}
where $x : \mathbb{R} \to \mathbb{R}^n$, $A$, $B$, $C$ are real matrices of suitable dimensions, $x(0)=0$. 
\begin{prop}{\cite{Fischer:1998}}
Let the input signal $u(t)\in\l^2\lbrack 0,\infty)$ to system \eqref{eq:LTI} be defined by its Laguerre spectrum $\{u_j\}_{j\in\mathbb{N}}$. Then the output Laguerre spectrum  $\{y_j\}_{j\in\mathbb{N}}$ is given by the output of the system 
\begin{align}\label{eq:LTI_laguerre}
x_{j+1}&= Fx_j +Gu_j,  \\
y_j&= Hx_j+ {Ju_j}, \nonumber 
\end{align}
where 
\begin{align*}
F&= {(I - \sqrt{p}A)}^{-1}(A-\sqrt{p}I )\\ 
G&= (1-p){(I- \sqrt{p}A)}^{-1}B, \\
H&= C{(I- \sqrt{p}A)}^{-1}, \\
{J}& = { \sqrt{p}~C{(I- \sqrt{p}A)}^{-1}B},
\end{align*}
$I$ is the identity matrix, and $p$ is the Laguerre parameter.
\end{prop}
An important implication of the above result is that the system description of \eqref{eq:LTI_laguerre}  possesses a ``casuality" property in   Laguerre-domain; The throughput term expressed by the matrix $J$ directly relating the input coefficient $u_j$ to the output coefficient $y_j$ is always present but the coefficients of higher order, i.e. $u_k, j<k$ do not contribute to the value of $y_j$. This is despite the fact that each Laguerre coefficient is, according to \eqref{eq:series}, evaluated from the whole signal sequence defined on $\mathbb{N}_0$. This property of the Laguerre-domain description will be exploited in Section~\ref{sec:estimation}.

\subsection{Discrete delay in Laguerre-domain}
Consider now a noise-free case of \eqref{eq:ddelay}, i.e. let $e(t)\equiv 0$
\begin{equation}\label{eq:delay}
y(t)= u(t-\tau).
\end{equation}
Then the operation of the delay block in Laguerre-domain is readily described by the following result.
\begin{prop}[\cite{MBU20}]\label{prop:delay_Laguerre}
Let the input and output signals of \eqref{eq:delay} be 
\[ U(z)=\sum_{k=0}^\infty u_k L_k(z), \quad Y(z)=\sum_{k=0}^\infty y_k L_k(z).
\]
Then the Laguerre spectrum of the output is related to that of the input by 
\begin{align} \label{eq:delay-coeffs}
    y_j = (1-p)\sum_{k=0}^{j-1}
    L_{j-k}^{(\tau)}(\sqrt{p})u_k + \sqrt{p}^\tau u_j,
    \end{align}
where
\[
L_m^{(\tau)}(\sqrt{p}) = (-\sqrt{p})^{m-\tau}\sum_{n=0}^{\tau-1}\binom{m+n}{n}\binom{m-1}{\tau-n-1}(-p)^n,
\]
and it is agreed that $\binom{n}{k}=0$ for $k>n$ by definition.
\end{prop}

Naturally, the discrete delay operator is an LTI system and can be written in state-space form \eqref{eq:LTI}. Then the delay length $\tau$ becomes the order of the state-space representation. By transforming \eqref{eq:delay} to Laguerre-domain, the delay estimation problem can be formulated as a parameter estimation problem and solved in a closed form, see \cite{M22} for details.

The role of the polynomials $L_m^{(\tau)}$ in the Laguerre-domain description of the delay operator is revealed by the convolution operator in \eqref{eq:delay-coeffs}. Indeed, the following relationship holds
\begin{equation}\label{eq:markov}
 (1-p) L_j^{(\tau)}(\sqrt{p})= h_j \triangleq  H_\tau F_\tau^{j-1}G_\tau, \quad j=1,2,\dots,
 \end{equation}
where $H_\tau, F_\tau, G_\tau$ are the matrices  given by \eqref{eq:LTI_laguerre} and evaluated for  delay operator \eqref{eq:delay}.

\section{Measurement noise in Laguerre-domain}\label{sec:noise}
There is no established noise model for Laguerre-domain system representations. Below, three  models are analyzed. The first two are conventional and defined in  time-domain, while the third is novel and introduced directly in Laguerre-domain.

\subsection{White noise}
The white noise model is typically used in communication systems to represent channel noise, in radar/sonar to describe the ambient noise, and to capture electronic noise  in solid-state electronics \cite{kay2013fundamentals}. 

Define the time-domain noise vector
\begin{equation}\label{eq:time_domain_noise}
\Et \triangleq   \begin{bmatrix}e(0) & e(1) & \dots &e(T-1)\end{bmatrix}^\intercal,
\end{equation}
and denote its covariance matrix as $\St$. Then the Laguerre spectrum, which is also further on referred to as (spectrum) distortion, is given by
\[
\begin{aligned}
\El & =\begin{bmatrix}e_0 & e_1 & \dots &e_{L-1}\end{bmatrix}^\intercal \triangleq \Psi_p \Et,
\end{aligned}
\]
where
\[
\Psi_p = {\left(\Phi_L^\intercal(p) \Phi_L(p)\right)}^{-1}\Phi_L^\intercal(p)
\]
is the projection matrix onto the space spanned by the first $L$ Laguerre functions, and the  matrix $\Phi_L(p)\in \mathbb{R}^{T\times L}$ contains the first $T$ instants of the Laguerre functions in time-domain. 
We will now show that when $e(t)$ is stationary and white, the entries of $\El$ are approximately uncorrelated with approximately the same variance as that of the time-domain noise. This is due to the orthonormality of the Laguerre functions.

\medskip

\begin{prop}\label{pr:white}
Suppose that  the time-domain noise sequence $\{e(t), t =  0, \dots, T-1\}$ is uncorrelated with zero mean and variance $\lambda$. Then, for a sufficiently large $T$,  the sequence of the Laguerre-domain coefficients $\{e_k, = 0,\dots, L-1 \}$ has arbitrary small correlations, and for each $k$ the variance $\mathbb{E}[e_k^2] = \lambda_L \approx \lambda$. 
\end{prop}

\begin{proof} For a finite $T$,  $\Et$ is square-summable and, therefore, can be described by an infinite Laguerre series (cf. \eqref{eq:series}), due to completeness of the basis. 
By assumption,  $\St = \lambda I_T$. Thus, by definition $\Sl(p) = \lambda \Psi_p \Psi_p^\intercal = \lambda (\Phi_L(p)^\intercal\Phi_L(p))^{-1}$.
The $(k,l)$-th entry of $\Phi_L(p)^\intercal\Phi_L(p)$ is given by
$\sum_{t=0}^{T-1} \ell_k(t;p)\ell_l(t;p)$ which converges, as $T\to\infty$, to 0 if $k\neq l$ and to 1 if $k=l$. By the continuity of the inverse operator, it holds that, for any arbitrary small $\epsilon >0$ and any $p$, there exists a sufficiently large $T$ such that $\|\Psi_p \Psi_p^T - I_L \| < \epsilon$
\end{proof}

\medskip

The  assumption of  noise stationarity is significant here:
When the time-domain noise is uncorrelated but non-stationary, i.e. has time-varying variances $\lambda_t$, the sequence of the Laguerre-domain coefficients $e_0, \dots, e_{L-1}$ is correlated. This can be easily seen by observing that when $\St = \Lambda$, a diagonal matrix whose $t$-th diagonal entry is $\lambda_t$, the covariance $\Sl(p) = \Psi_p \Lambda\Psi_p^\intercal = (\Phi_L^\intercal\Phi_L)^{-\intercal}\Phi_L^\intercal\Lambda\Phi_L(\Phi_L^\intercal\Phi_L)^{-\intercal}$. Notice that the $(k,l)$-th entry of the matrix $\Phi_L^T\Lambda\Phi_L$ in this case is 
\[
[\Phi_L^T\Lambda_t\Phi_L]_{kl} = \sum_{t=0}^{T-1} \lambda_t \ell_k(t;p) \ell_l(t;p).
\]
A main conclusion of Proposition~\ref{pr:white} is that, for a white $e(t)$, the second-order properties of the Laguerre-domain distortion are independent of the Laguerre parameter $p$ for sufficiently large $T$.

\subsection{Colored noise}
Stationary correlated (colored) disturbances are usually modeled as filtered white noise. There are a variety of filter structures ranging from autoregressive filters to rational transfer operators that can be employed in the modeling; see e.g., \cite{Ljung1999}. Regardless of the used model structure, and under zero mean assumption, the second-order properties of colored noise are given by its correlation function. It can be used to construct a full covariance matrix $\St$ of the noise vector $\Et$. 

The Cholesky factorization of a known positive-definite covariance $\St$ is 
\[
\St = \lambda S S^\intercal,
\]
where $\lambda>0$ is the variance of $e(t)$ and $S$ is a lower-triangular matrix. The columns of $S$ are given by the impulse response of the spectral factor of the noise process.
Then it holds that
\[
\Sl(p) = \lambda \Psi_p S S^\intercal \Psi_p^\intercal = \lambda \tilde{\Psi}_p \tilde{\Psi}_p ^\intercal,
\]
and a convolution between the basis functions and the impulse response of the spectral factor of the noise takes place. In general, $\tilde{\Psi}_p \tilde{\Psi}_p ^T$ is a full matrix and does not converge to the identity matrix as $T\to\infty$. Therefore, in contrast with the case of white noise, the second-order properties of $\El$  depend on $p$, and, in general, the entries of $\El$ will be correlated. This is demonstrated in Fig.~\ref{fig:cov_laguerre_correlated}

\begin{figure}
    \centering
    \includegraphics[width=0.45\textwidth]{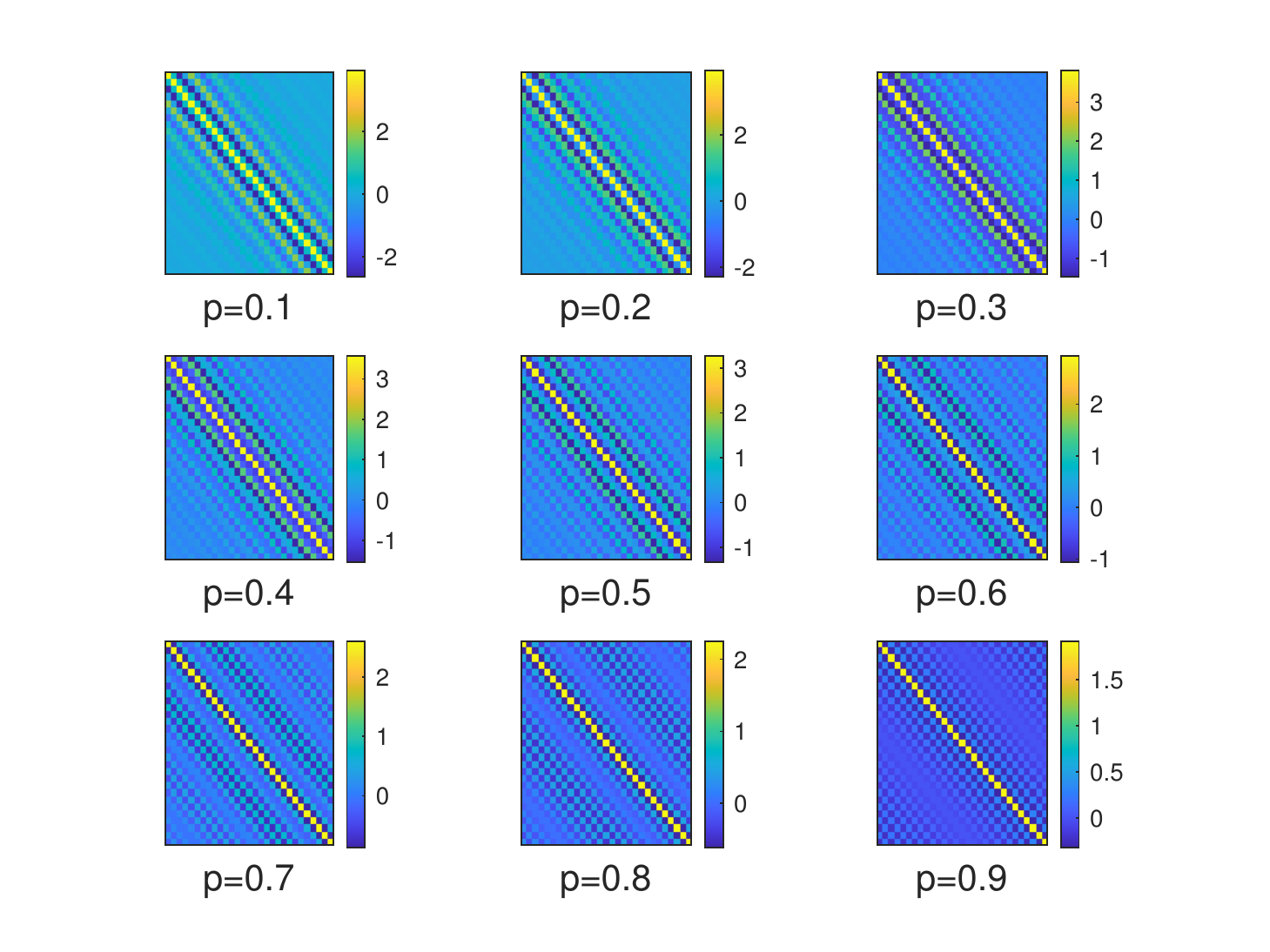}
    \caption{Covariance matrix of the Laguerre distortion $\El$ for values of the Laguerre parameter $p\in \lbrack 0.1, 0.9\rbrack$. Here $L=29$ and $T=1500$. The time-domain noise vector is a subsequence of a unit variance stationary white noise filtered through a second-order transfer function with unit gain, complex poles $0.4732 \pm 0.7190i$ and no zeros. }
    \label{fig:cov_laguerre_correlated}
\end{figure}

\subsection{Random combination of Laguerre functions}
 When no assumptions are made regarding the realizations of $e(t)$, the measured signal $\{y(t), t\in \lbrack 0, T-1 \rbrack\}, T\to\infty$, may not be in $\l^2\lbrack0,\infty)$.
 It will be almost surely in $\l^2$ only if the realizations of $e(t)$ are almost surely in $\l^2$. This property can be guaranteed  if, for instance, $e(t)$ is defined as a stochastic process given by a random combination of a finite number of Laguerre functions
\begin{equation}\label{eq:random_Laguerre_process}
e(t) \triangleq \sum_{k=0}^K \mathrm{e}_k \ell_k(t; p_e), 
\end{equation}
where $K\in \mathbb{N}_+$, $0<p_e<1$, and $\mathrm{e}_0, \dots, \mathrm{e}_K$ are random variables with zero mean and finite variances. This is in contrast with \eqref{eq:series}, where the Laguerre coefficients are constant.
Model \eqref{eq:random_Laguerre_process} is related to what is employed in the Karhunen–Lo{\`e}ve expansion (\cite[Ch.3, Sec. 4]{wong2012stochastic}). Both models decouple the probabilistic behavior of the signal from its behavior in time.   Yet, the models differ in character: The Karhunen–Lo{\`e}ve expansion appears in a representation theorem applicable to an infinite sum and requires finite support of the orthonormal functions, while (\ref{eq:random_Laguerre_process}) is a defining model. Another important difference is that in (\ref{eq:random_Laguerre_process}) the random variables $\mathrm{e}_k$ are not necessarily  uncorrelated. 

Noise model \eqref{eq:random_Laguerre_process} produces a non-stationary second-order process, which comes in contrast to the stationary models usually used to model noise in time-domain. Its correlation function is given by
\[
\begin{aligned}
R_e(t,s) = \mathbb{E}[e(s)e(t)] &= \sum_{m=0}^K\sum_{n=0}^K \sigma_{mn} \ell_m(s; p_e) \ell_n(t; p_e),\\
\end{aligned}
\]
where $\sigma_{mn}$ is the correlation between $\mathrm{e}_m$ and $\mathrm{e}_n$. For mutually uncorrelated coefficients $\mathrm{e}_k$, this expression reduces to
\[
\begin{aligned}
R_e(t,s) = \sum_{k=0}^K \mathbb{E}[\mathrm{e}_k^2]\ell_k(s; p_e) \ell_k(t; p_e).
\end{aligned}
\]

 Notably,  model \eqref{eq:random_Laguerre_process} is able to capture projections of both white and colored time-domain finite  noise subsequences  onto the space spanned by the first $K+1$ Laguerre functions. In other words, in this subspace and with the freedom to model the correlations matrix of $\El$, it is sufficient to assume model \eqref{eq:random_Laguerre_process} irrespective of the actual time-domain correlation properties. 
 
 The noise models are illustrated in Fig.~\ref{fig:real_lag_model}, and also utilized  in Section~\ref{sec:delay_estimation}.

\begin{figure}
    \centering
    \includegraphics[width=0.47\textwidth]{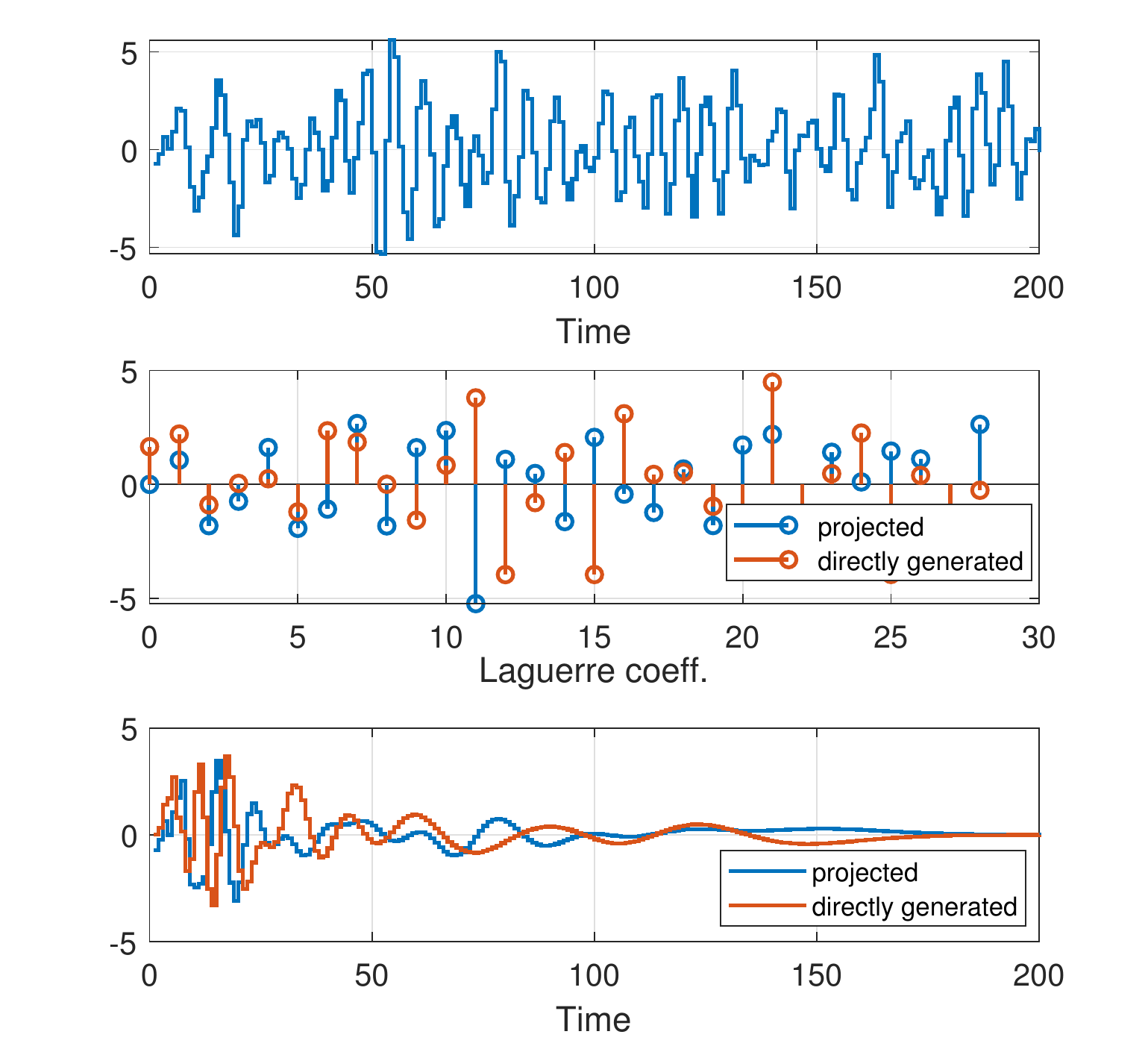}
    \caption{Top panel: time-domain colored noise. Middle panel: corresponding Laguerre-domain distortion (blue) and a random  distortion directly generated in Laguerre-domain ($p=0.5$). Lower panel: reconstructed time-domain noise using coefficients from the middle panel (model (\ref{eq:random_Laguerre_process}), $p_e = p$). The difference between the blue signals in the top and lower panels is the projection residual.} 
    \label{fig:real_lag_model}
\end{figure}

\section{Noise-corrupted Laguerre spectrum}
Introduce the vectors 
\[
\begin{aligned}
\Yt &\triangleq
\begin{bmatrix}
y(0)&y(1)& \dots& y(T-1)
\end{bmatrix}^\intercal, \\
\Yl &\triangleq
\begin{bmatrix}
y_0& y_1& \dots &y_{L-1}
\end{bmatrix}^\intercal.
\end{aligned}
\]
Then, in terms of the noise-free output $\Yt$, an approximation of the first $L$ Laguerre spectrum coefficients $\Yl$ from the data in the interval $\lbrack 0, T-1\rbrack$ is given by
\[
\Ylb = \Psi_p \Yt.
\]
The errors $\mathcal{E} = \Yl - \Ylb$ are truncation errors (residuals), due to the contribution from the Laguerre coefficients of  orders higher than $L$. Notice that the dependence of this error on $p$ is implicit in the notation.

When the measurement is noisy, only  $Y_\text{meas.} = \Yt + \Et$ is available, where $\Et$ is defined in (\ref{eq:time_domain_noise}) and $Y_\text{meas.}$ is a vector stacking the measured outputs. This introduces further errors, and the approximation becomes
\begin{equation}\label{eq:noisy-laguerre-spectrum}
\begin{aligned}
\!\!\!\!\! \Ylh \triangleq \Psi_p Y_\text{meas.} =\Psi_p (\Yt+\Et) &= \Ylb + \Psi_p \Et\\
& = \Ylb + \El\\
&= \Yl + \mathcal{E} + \El,
\end{aligned}
\end{equation}
where $\El$ represents the distortion in the first $L$ Laguerre spectrum coefficients of $y(t)$, due to the measurement noise, and its covariance matrix is given as
\begin{equation}\label{eq:lag_noise_cov}
\Sl = \mathbb{E}[\El\El^\intercal] = \Psi_p \St \Psi_p^\intercal .
\end{equation}
From (\ref{eq:delay-coeffs}), the following relation then holds
\[
\Ylh = M_\tau(p) \Ul + \mathcal{E} + \El,
\]
where $\Ul$ is a column vector stacking the Laguerre spectrum coefficients of the input $u(t)$ (which are usually known by design), and $M_\tau(p)$ is the matrix of Markov parameters. In other words,
\begin{equation}\label{eq:noisy-delay-coeffs}
\begin{aligned}
\hat{y}_j&= (1-p)\sum_{k=0}^{j-1}
    L_{j-k}^{(\tau)}(\sqrt{p})u_k + \sqrt{p}^\tau u_j+ \epsilon_j + e_j, \\
    & \text{in which } j =0, \dots, L-1.
\end{aligned}
\end{equation}
By the completeness of the Laguerre basis in $\l^2[0,\infty)$, the norm of the truncation error goes to zero as $L$ grows towards infinity. Notice that the truncation is not a problem when the input spectrum is finite. Even with a significant ``tail" in the noise-free output, all the relations hold if the output spectrum is evaluated from a complete realization. However, unlike the truncation error, the distortion due to measurement noise persists even when $L$ and $T$ are large. In what follows, we will assume that $\epsilon_j = 0$ for all $j$.

\section{Signal estimation in Laguerre-domain}\label{sec:estimation}

An important consequence of the fact that a state-space time-domain LTI model equivalently translates into a state-space LTI description in Laguerre-domain is that ``causality" also applies to the latter, \cite{NY82}. Since the Laguerre spectrum of a signal is calculated from its time-domain values on $\lbrack 0, \infty)$, causality does not hold in a temporal sense but rather means that the Laguerre coefficients $y_j$ in \eqref{eq:LTI_laguerre} are independent of $u_i, i>j$.

Consider now \eqref{eq:LTI} with output measurements corrupted by additive noise. Assuming that the input signal is formed so that $u_i=0, i=0,\ldots,m-1$, for some $m\in \mathbb{N}_+$, the first $m$ coefficients of the output, i.e. $y_i, i=0,\ldots,m-1$, are independent of the input and constitute instead the first $m$ coefficients of the Laguerre spectrum  of the noise realization. Thus, the signal shape of the realization can be reconstructed and utilized for noise reduction. 
The accuracy of this reconstruction depends on how complex the signal shape of the disturbance realization is, that is how many Laguerre coefficients in a truncated Laguerre series of the signal it takes to achieve the desired result. Further, when  the Laguerre noise spectrum is random and constitutes a correlated sequence, the noise realization coefficients of higher order than $m$ can be predicted and enhance noise reduction even more

\subsection{Noise vector estimation}
Consider the particular case of (\ref{eq:noisy-delay-coeffs}) where the input is delayed in the Laguerre-domain. That is, the input signal is designed in the Laguerre-domain so that, for some $m \in \mathbb{N}_+$, it holds that $u_n = 0$ for all $n<m$
\begin{equation}\label{eq:delayed_Laguerre}
\!\!\!\!\hat{y}_j\!=\! \begin{cases} 
                         e_j, \quad j<m,\\
                         \sqrt{p}^\tau u_j + e_j, \quad j = m,\\
                        (1-p)\sum_{k=m}^{j-1}
                         L_{j-k}^{(\tau)}(\sqrt{p})u_k + \sqrt{p}^\tau u_j  + e_j, \; j> m.
            \end{cases}\!\!
\end{equation}
Then $\hat{y}_0, \dots, \hat{y}_{m-1}$ are equal to the distortion coefficients $e_0, \dots, e_{m-1}$, respectively. 

Let the matrix $\Xi_p$ be the submatrix of $\Psi_p$ given by its first $m$ rows. Then,
\[
\begin{bmatrix}
e_0 & \hdots & e_{m-1}
\end{bmatrix}^\intercal = \Xi_p Y_\text{meas.}
\]
An approximation of the \textit{whole} time-domain noise vector can then be obtained as
\begin{equation}\label{eq:noise_estimate}
\doublehat{E}_\text{time} = \begin{bmatrix}
\doublehat{e}(0) & \hdots & \doublehat{e}(T-1)
\end{bmatrix}^\intercal \triangleq\Phi_{m}(p)\Xi_p Y_\text{meas.}
\end{equation}
In other words, for $ t= 0, \dots, T-1$,
\[
\doublehat{e}(t) = \sum_{k=0}^{m-1}e_k \ell_k(t; p) = \sum_{k=0}^{m-1}\; [\Xi_p]_{k:}\, Y_\text{meas.}\; \ell_k(t; p) , 
\]
where $[\Xi_p]_{k:}$ denotes the $k$-th row of $\Xi_p$.

\subsection{Best linear estimate}
When the covariance matrix of the distortion vector $\El$ is non-diagonal and known, the best linear estimator (BLE) of $e_m, \dots, e_{L-1}$ in terms of $e_0, \dots, e_{m-1}$ can be readily computed. If, in addition, the measurement noise is Gaussian, the  estimator obtained is in fact mean squared error optimal. 

Partition the vector $\Ylh$ and the covariance matrix $\Sl$ as
\[
\Ylh =\begin{bmatrix} \Ylh^{(1)} \\ \Ylh^{(2)}\end{bmatrix}, \quad  \El =\begin{bmatrix} \El^{(1)} \\ \El^{(2)}\end{bmatrix}, \quad
\Sl = \begin{bmatrix} \Sl^{(11)} & \Sl^{(12)}\\ 
\Sl^{(21)} & \Sl^{(22)}\end{bmatrix},
\]
such that $\Ylh^{(1)}$ and $\El^{(1)}$  contain the first $m$ entries of their corresponding vectors, and $\Sl^{(11)}$ is the covariance matrix of $\El^{(1)}$; it is square and of size $m$. From (\ref{eq:delayed_Laguerre}), $\Ylh^{(1)}  = \El^{(1)}$, and the BLE of $\El^{(2)}$ is given as
\[
\begin{aligned}
\Elh^{(2)} =\begin{bmatrix}
\hat{e}_{m} & \hdots & \hat{e}_{L-1}
\end{bmatrix}^\intercal &= \Sl^{(21)}\left[\Sl^{(11)}\right]^{-1}\Ylh^{(1)}\\
&= \Sl^{(21)}\left[\Sl^{(11)}\right]^{-1} \Xi_p Y_\text{meas.}
\end{aligned}
\]
It has a covariance
\[
\text{\bf cov}\left( \Elh^{(2)} \right) = \Sl^{(21)}\left[\Sl^{(11)}\right]^{-1}\Sl^{(12)}
\]
and the covariance of the error is
\begin{equation}\label{eq:cov_ble}
\text{\bf cov}\left( \El^{(2)} - \Elh^{(2)} \right) = \Sl^{(22)} -  \Sl^{(21)}\left[\Sl^{(11)}\right]^{-1}\Sl^{(12)}.
\end{equation}
Then an estimate of the time-domain noise is given as
\[
\hat{e}(t) = \sum_{k=0}^{m-1}e_k \ell_k(t; p) + \sum_{k=m}^{L-1} \hat{e}_k \ell_k(t; p), \quad t = 0, \dots, T-1.
\]
or, equivalently,
\begin{equation}\label{eq:noise_estimate_mms}
\begin{aligned}
\Eth &=\Phi_L(p) \Elh =  \Phi_L(p) \begin{bmatrix}
\El^{(1)} & 
\El^{(2)} \end{bmatrix}^\intercal\\ 
&= \Phi_L(p)\begin{bmatrix} 
\Xi_p\\
\Sl^{(21)}\left[\Sl^{(11)}\right]^{-1} \Xi_p 
\end{bmatrix}Y_\text{meas.}\\
& = \doublehat{E}_\text{time} + \Theta_p Y_\text{meas.},
\end{aligned}
\end{equation}
where $\Theta_p = \Phi_{m+1:L}(p)\Sl^{(21)}\left[\Sl^{(11)}\right]^{-1} \Xi_p$ and $\Phi_{m+1:L}(p)$ denotes the last $L-m$ columns of $\Phi_{L}(p)$.
This is to be compared to (\ref{eq:noise_estimate}) where the coefficients $e_{m}, \dots, e_{L-1}$ are estimated using their (unconditional) expected value; namely zero. 

\subsection{Noise reduction}\label{sec:noise_reduction}
Suppose that an estimate  $\doublehat{E}_\text{time}$ is obtained according to  (\ref{eq:noise_estimate}). Then subtract the noise estimate from the measurement vector; viz. $\doublehat{Y}_\text{meas.} =  Y_\text{meas.} - \doublehat{E}_\text{time}$. 
Using (\ref{eq:noise_estimate}) and recalling that $Y_\text{meas.} = \Yt + \Et$ ,
\[
\begin{aligned}
\doublehat{Y}_\text{lag.} &= \Psi_p( Y_\text{meas.} - \doublehat{E}_\text{time}) = \Ylh -\Psi_p \doublehat{E}_\text{time}\\
&=\Psi_p \Yt + \Psi_p (I_T - \Phi_{m}(p)\Xi_p) \Et  = \Yl + \doublehat{E}_\text{lag.},
\end{aligned}
\]
where the equality before the last one holds because $\Xi_p\Yt = 0_{m\times1}$. Here, the  truncation errors are ignored, i.e., it is assumed that $\Yl = \Ylb$. Then it holds, $\Psi_p (I_T - \Phi_{m}(p)\Xi_p) \Et = \begin{bmatrix} 0 & \El^{(2)} \end{bmatrix}^\intercal$. Comparing to \eqref{eq:noisy-laguerre-spectrum}, where the distortion due to noise has a covariance as in (\ref{eq:lag_noise_cov}), the covariance matrix  is 
\[
\Psi_p (I_T - \Phi_{m}(p)\Xi_p) \St (I_T - \Phi_{m}(p)\Xi_p)^\intercal \Psi_p^\intercal .
\]
It is not difficult to see that this matrix is exactly equal to the submatrix of $\Sl$, in (\ref{eq:lag_noise_cov}), given by the last $L-m$ rows and columns
\[
\begin{bmatrix}
0 & 0\\
0 & I_{L-m}
\end{bmatrix} \Sl \begin{bmatrix}
0 & 0\\
0 & I_{L-m}
\end{bmatrix} ,
\]
where  $0$ denotes zero matrices of appropriate dimensions, and $I_{L-m}$ is the identity matrix with dimension $L-m$.
Thus,  no improvement in the signal-to-noise ratio in the Laguerre-domain is obtained (the signal, i.e., non-zero spectrum, starts at $m$). 

\subsection*{Noise reduction using BLE in Laguerre-domain}
Now suppose that the BLE $\Eth$ (\ref{eq:noise_estimate_mms}) is used instead. Then
\[
\begin{aligned}
{\Ylhs} &= \Psi_p( Y_\text{meas.} - \Elh) = \Ylh -\Psi_p \doublehat{E}_\text{time}\\
&=\Psi_p \Yt + \Psi_p (I_T - \Phi_{m}(p)\Xi_p) \Et - \Psi_p \Theta_p \Et \\
&= \Yl + \doublehat{E}_\text{lag.}- \Psi_p \Theta_p \Et\\
& = \Yl + \begin{bmatrix} 0 & \El^{(2)} - \Elh^{(2)}  \end{bmatrix}^\intercal.
\end{aligned}
\]
In the light of (\ref{eq:cov_ble}),  a reduction in the noise variance of ${\Ylhs}$ compared to $\doublehat{Y}_\text{lag.}$ or $\Ylh$ is guaranteed. 

To recapitulate the above results,  no noise reduction can be achieved in the Laguerre-domain if the distortion $\El$ is white. When, however, the distortion is correlated with a known correlation function, then the BLE can be used to improve the signal-to-noise ratio in the Laguerre-domain.

\section{Delay Estimation}\label{sec:delay_estimation}
In the rest of the paper, the results from Section~\ref{sec:estimation} are applied to the problem of delay estimation in Laguerre-domain. Whereas a discrete-time delay estimation problem is essentially system order estimation in time-domain, it can be formulated as a parameter estimation problem in Laguerre-domain. An additional argument for the use of Laguerre-domain is that the input in delay estimation applications is typically a finite pulse of a certain signal shape and readily lends itself to a Laguerre series representation.
\subsection{Algorithm}
A result from an earlier contribution \cite{M22} showed that the value of $\tau$ can be computed using three subsequent Markov parameters, $h_{m-1}, h_m$,  and $h_{m-1}$ in \eqref{eq:markov} for any value of $m\geq n+1$ using the formula
\[
\tau = - \frac{(m+1) h_{m+1} + (m-1) h_{m-1}}{ \beta h_m} - \frac{m\alpha}{\beta},
\]
where $n$ is the index of the first non-zero Markov parameter, $\alpha =  (\sqrt{p} + \sqrt{p}^{-1})$, $\beta =  (\sqrt{p} - \sqrt{p}^{-1})$, $h_k = \sum_{j=1}^k g_{k-j}y_j$,  and 
\[
g_n = \frac{1}{u_n}, \qquad g_{k} = -\frac{1}{u_n} \sum_{j=0}^{k-1} u_{k-j}g_j, \quad k\geq n+1, \quad u_n \neq 0.
\]
From this, and using the first $M$ non-zero Markov parameters, it is straightforward to see that the following equality holds
\[
a- b \tau =0,
\]
in which the column vectors $a$ and $b$ are defined as
\[
a \triangleq\Omega(\alpha) \begin{bmatrix}
h_{n+1}\\
\vdots\\
h_{M+n-1}
\end{bmatrix} + (M+n) \begin{bmatrix}
0\\
\vdots\\
0\\
h_{M+n}
\end{bmatrix}, \;\; \triangleq \beta \begin{bmatrix}
h_{n+1}\\
\vdots\\
h_{M+n-1}
\end{bmatrix},
\]
and $\Omega(\alpha)$ is the tridiagonal matrix
\[
\Omega(\alpha) \triangleq \begin{bmatrix}
\alpha & 2 & 0  & 0& \dots & 0 &0 \\
1 & 2\alpha & 3 & 0 &  \dots & 0& 0\\
0 & 2 & 3\alpha & 4 & \hdots & 0&0 \\
\vdots &  &  & & \ddots & \vdots & \vdots\\
0 &  0&  0 & 0& \hdots & (M-2)\alpha & (M-1)\\
0 & 0 & 0 & 0 &\hdots& (M-2) & (M-1)\alpha\\
\end{bmatrix}.
\]
Hence,  the delay is given by the closed-form formula for $$\tau = -\frac{b^\intercal a}{b^\intercal b}.$$
When the  Markov parameters are estimated $\{\hat{h}_k\}$ from noisy data, e.g., using (\ref{eq:noisy-delay-coeffs}) and ordinary least-squares, the following estimate is obtained
\begin{equation}\label{eq:delay_formula}
\hat{\tau} = -\frac{\hat{b}^\intercal \hat{a}}{\hat{b}^\intercal \hat{b}}.
\end{equation}

\subsection{Numerical Experiment}

A Monte Carlo numerical experiment is presented below to illustrate the performance of estimate \eqref{eq:delay_formula}. Three data sets are considered that correspond to the three noise models detailed in Section~\ref{sec:noise}. 

Let the true delay  value in (\ref{eq:ddelay}) be $\tau = 4$, and consider the following three noise models (NM)
\smallskip
\begin{itemize}
     \setlength{\itemindent}{0.5em} 
    \item[{NM1:}] $e_1(t)$ is a stationary Gaussian white noise with variance $\lambda = 0.3$;
    \item[{NM2:}] $e_2(t)$ is a stationary colored noise defined as
    \[
    e_2(t) = \frac{1}{q^2-0.9464 q + 0.7408} v(t),
    \]
    where $q$ is the time-domain shift operator, and $v(t)$ is a Gaussian white noise whose variance is such that the variance of $e_2(t)$ is, as $e_1(t)$, equal to $\lambda = 0.3$;
    \item[{NM3:}] $e_3(t)$ is given by (\ref{eq:random_Laguerre_process}) with $p_e = 0.5$, $K = 19$, such that $\text{\bf cov}(\begin{bmatrix} e_3(0) \dots e_3(T-1) \end{bmatrix}) = \text{\bf cov}(\begin{bmatrix} e_2(0) \dots e_2(T-1) \end{bmatrix}) $.
\end{itemize}
\smallskip
 In Laguerre-domain and,  for sufficiently large $T$, all the three models have the same marginal second-order properties, and the last two produce Laguerre-domain noise vectors $\El$ with the same covariance matrix. Yet the time-domain properties of {NM2} and {NM3} are different given that {NM3}  is non-stationary.

We ran a Monte Carlo simulation experiment for the three noise models and  the total number of simulations was $15\mathrm{e}^5$. The time-domain data set length was 300 samples, and all computations were made using the first $L = 20$ Laguerre functions with a Laguerre parameter $p=0.5$. The input was designed  to possess the Laguerre spectrum
\[
\Ul = [\underbrace{0 \dots 0}_{\text{15 zeros}} \;\; 3.1 \;\; 3 \;\; 0 \;\; 0\;\; 0]^\intercal,
\]
and kept fixed during all simulations. Thus, the first $15$ Laguerre coefficients of the noise-free outputs are identically zero. Only the coefficients from $y_{15}$ to $y_{19}$ were used to compute the least-squares estimate of the first $5$ Markov parameters yielding a delay estimate according to  \eqref{eq:delay_formula}.

The results of the experiment for each case is summarized in terms of the mean and the variance of the estimator in Table \ref{tab:results_without_noise_reduction}. The variances for the colored noise  cases  are larger than that for the white noise case. This is natural as all the noise models have the same marginal variance.

\begin{table}[!ht]  
  \centering
  \begin{tabular}{ccccc}
    \toprule
            & NM1       & NM2       & NM3     & $\tau$   \\ \midrule
    Mean    & 3.3807    & 3.2229    & 3.2234   &4  \\
    Var     & 0.8904    & 1.0827    & 1.0839   &  \\ \bottomrule
  \end{tabular}
    \caption{Mean and variance of $\hat{\tau}$ for the three noise models. }
  \label{tab:results_without_noise_reduction}
\end{table}

Table~\ref{tab:results_with_noise_reduction} shows the results obtained when the BLE is used to reduce the noise, as pointed out in Section \ref{sec:noise_reduction}, using the same exact data sets. As expected, there is no improvement in the case of the white noise model. However, the improvement in the mean value as well as the variance of $\hat{\tau}$ is clear in case of NM2 and NM3: the bias is reduced from $0.7771$ to $0.3080$ (about 60\% drop), and the variance is reduced from $1.0827$ to $0.5918$  (about 45\% drop). Interestingly, in Table~\ref{tab:results_with_noise_reduction}, the mean value rounds up to the true delay value. 
We also note that the results obtained using NM2 and NM3 are almost identical, despite NM3 being non-stationary. This is expected as was indicated earlier. 
Fig.~\ref{fig:estimated_noise_in_laguerre} and Fig.~\ref{fig:estimated_noise_in_time} show a realization of the BLE of the Laguerre distortion vector, and the reconstructed time-domain noise, respectively. Clearly, the signal shape of the noise realization is reconstructed closely only in case of NM3, for which Parseval's identity holds.

\begin{table}[!ht]  
  \centering
 \begin{tabular}{ccccc}
    \toprule
            & NM1       & NM2           & NM3         &$\tau$   \\ \midrule
    Mean    & 3.3807    & \bf 3.6920    & \bf 3.6925   &4  \\
    Var     & 0.8904    & \bf 0.5918    & \bf 0.5919   &  \\ \bottomrule
  \end{tabular}
  \caption{Mean and variance of $\hat{\tau}$ for the three noise model after noise reduction via the BLE. Both the bias and the variance are improved compared to the results in Table~\ref{tab:results_without_noise_reduction}. }
  \label{tab:results_with_noise_reduction}
\end{table}

\begin{figure}[h]
    \centering
    \includegraphics[width=0.47\textwidth]{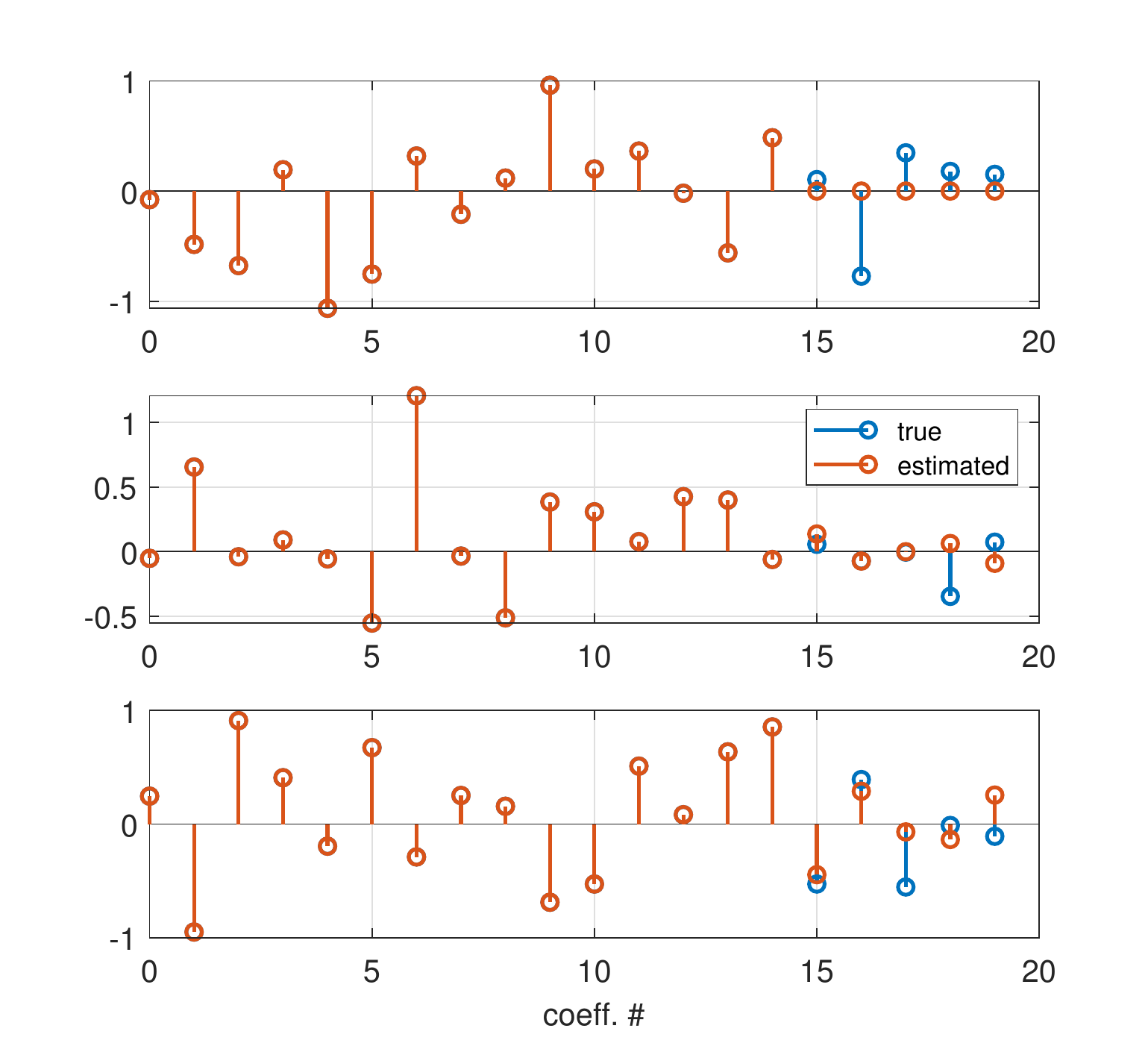}
    \caption{A realization of the BLE of the Laguerre distortion for the three noise models. Top panel: NM1, Middle panel: NM2, and Lower panel: NM3.  In the three cases, the "true" and 'estimated" first 15 coefficients coincide. For the case of NM1, the BLE of the last five coefficients is zero. The Laguerre-domain estimation errors (\(\|\El^{(2)} - \Elh^{(2)}\|^2\)) for these three realizations are: 0.7834 (NM1), 0.2009 (NM2), and 0.3993 (NM3). The corresponding noise vectors in time-domain are shown in Fig. \ref{fig:estimated_noise_in_time}. }
    \label{fig:estimated_noise_in_laguerre}
\end{figure}

\begin{figure}[h]
    \centering
    \includegraphics[width=0.47\textwidth]{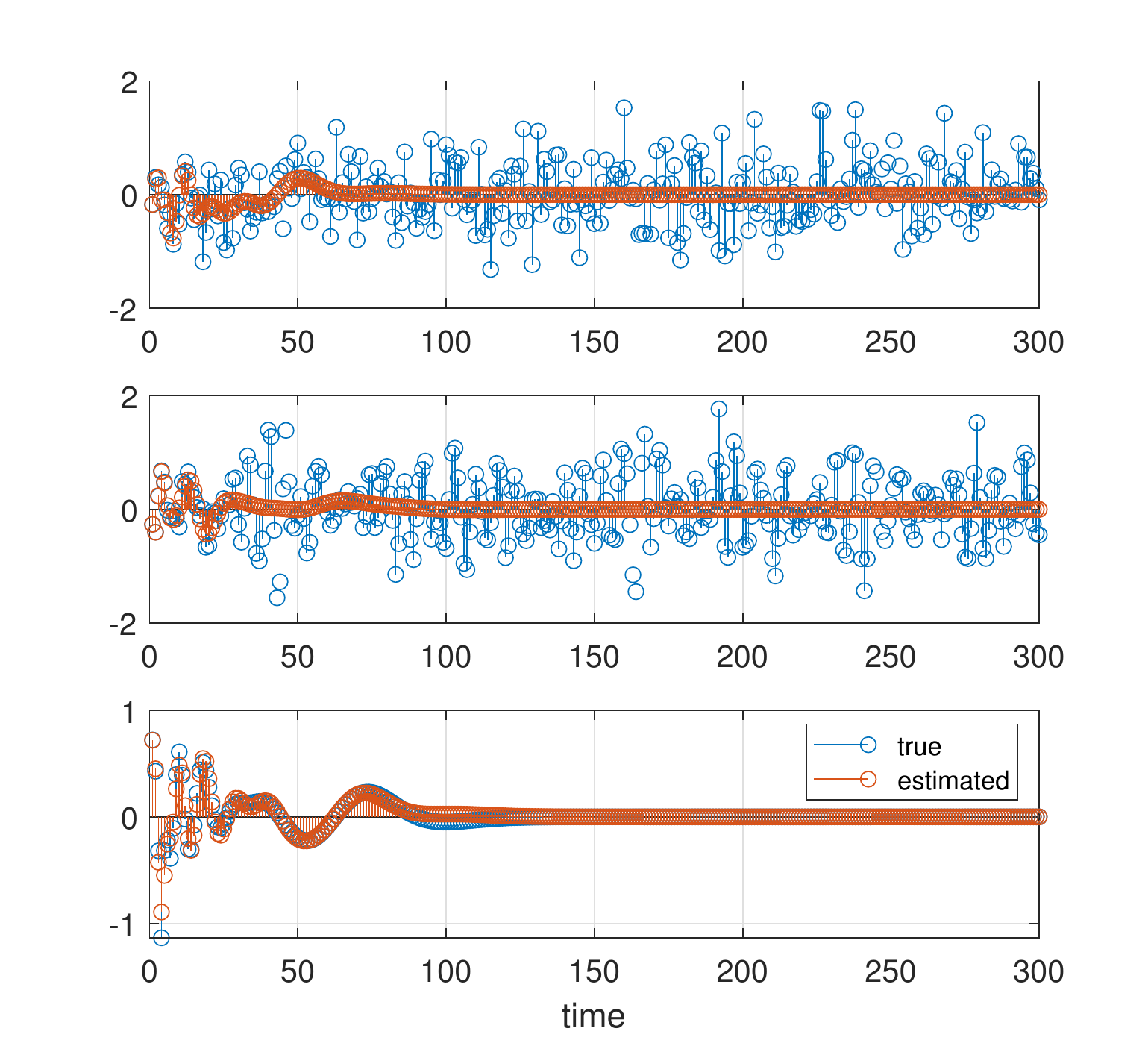}
    \caption{Examples of reconstructed time-domain noise vectors. Top panel: NM1, Middle panel: NM2, and Lower panel: NM3. The time-domain estimation errors (\(\|\Et - \Eth\|^2\)) for these three realizations are: 80.0112 (NM1), 92.6170 (NM2) and  0.3993 (NM3). 
    }
    \label{fig:estimated_noise_in_time}
\end{figure}

\section{Conclusions}
The implications of stochastic additive measurement noise on the accuracy of Laguerre-domain estimation are studied both analytically and via Monte-Carlo simulations. It is shown that, by selecting the input signal as a linear combination of higher-order Laguerre functions, the signal shape of the actual noise realization can be reconstructed from the spectrum of the output in the case of strongly correlated noise sequence. The reconstructed signal can be then used for noise reduction. The efficacy of the proposed approach is demonstrated with respect to a Laguerre-domain time delay estimation algorithm.

\bibliographystyle{IEEEtran}
\bibliography{bibliografi}

\end{document}